\documentclass{article}
\usepackage[english]{babel}
\usepackage[cp1251]{inputenc}
\usepackage{amssymb,amsmath,amsthm}
\newtheorem{thm}{Theorem}
\newtheorem*{thm*}{Theorem}
\newtheorem*{lem*}{Lemma}
\newtheorem{lem}{Lemma}
\newtheorem{cor}{Corollary}
\newtheorem{pr}{Proposition}

\begin{document}
\title{On solvability of a partial integral equation in the
space ${L_2(\Omega\times\Omega)}$}
\author{\sc{Eshkabilov Yu.Kh.}\\
\it{National University of Uzbekistan}\\
e-mail: yusup62@rambler.ru}
\date{}
\maketitle
\begin{abstract}

In this paper we investigate solvability of a partial integral equation in the space $L_2(\Omega\times\Omega),$
where $\Omega=[a,b]^\nu.$ We define a determinant for the partial integral equation as a continuous function on
$\Omega$ and for a continuous kernels of the partial integral equation we give explicit description of the
solution.


\end{abstract}

In the models of solid state physics [1] and also in the lattice field theory [2], there appear so called discrete
Schrodinger operators, which are lattice analogies of usual Schrodinger operators in continuous space. The study
of spectra of lattice Hamiltonians (that is discrete Schrodinger operators) is an important matter of mathematical
physics. Nevertheless, on studying spectral properties of discrete Schrodinger operators three appear partial
integral  equations in a Hilbert space of multi- variable functions [1,3]. Therefore, on the investigation of
spectra of Hamiltionians considered on a lattice, the study of a solvability problem for partial integral
equations in  $L_2$ is essential (and even interesting from the point of view of functional analysis).

A question on the existence of a solution of partial integral
equations for functions of two variables were considered in [4-8]
and others. In this paper we consider an integral equation on the
space of functions of two variables $L_2 (\Omega\times\Omega),$
where $\Omega=[a,b]^\nu\subset\mathbb R^\nu,$ with one partial
integral operator. We define a determinant for the partial
intergal equation (PIE) as a continuous function on $\Omega,$
which helps to obtain the classical Fredholm theorems for a PIE,
and for a continuous kernels of the PIE we give explicit
description of the solution.

Let ${\cal H}=L_2(\Omega\times\Omega)$  (${\cal H}_0=L_2(\Omega)$) be a Hilbert space of measurable and quadratic
integrable functions on $\Omega\times \Omega$ (on $\Omega$), where $\Omega=[a,b]^\nu.$ We denote by $\mu$ the
Lebesgue measure on $\Omega$ and define the measure $\widehat \mu$ on $\Omega\times\Omega$ by $\widehat
\mu=\mu\otimes\mu.$ In the space $\cal H$, we consider a partial integral operator (PIO) $T_1$ defined by
$$
T_1f=\int\limits_\Omega k(x,s,y)f(s,y)ds,\quad f\in{\cal H}
$$
where $k(x,s,y)\in L_2(\Omega^3).$ The function $k(x,s,y)$ is
called \emph{kernel} of the PIO $T_1.$

If there exists a number $M$ such that $$b(t)\le M \ \mbox{ for
almost all }\ t\in\Omega, \eqno (I)$$ then the operator $T_1$ is a
linear bounded operator on $\cal H$ and it is uniquely defined by
its kernel $k(x,s,y),$ where
$$
b(t)=\int\limits_\Omega \int\limits_\Omega |k(x,s,t)|^2dxds.
$$

A kernel $\overline{k(s,x,y)}$ corresponds to the adjoint operator $T_1^*,$ i.e.
$$
T_1^*f=\int\limits_\Omega\overline{k(s,x,y)}f(s,y)ds,\quad f\in
\cal H.
$$

Consider a family of operators $\{ K_\alpha\}_{\alpha\in \Omega}$
in ${\cal H}_0$ associated with $T_1$ by the following formula
$$
K_\alpha\varphi=\int\limits_\Omega k(x,s,\alpha)\varphi(s)ds,\quad
\varphi\in{\cal H}_0,
$$
where $k(x,s,y)$ is the kernel of $T_1$.

Further, if a set of integralabity in the integral is absent, then
we mean integralabity by the set $\Omega.$ First, we consider
certain properties of PIO $T_1$ with the kernel $k(x,s,y) \in
L_2(\Omega^3)$ satisfying the condition (I) and then we study
solvability of the PIE with the kernel $k(x,s,y)\in C(\Omega^3).$
\begin{lem} Let $f\in\cal H$ and $\varphi_y(x)=f(x,y),$  where $y\in\Omega$ is fixed. Then for an arbitrary
$\varepsilon>0,$ there exists a subset $\Omega_\varepsilon\subset
\Omega$ such that $\mu(\Omega_\varepsilon)\ge \mu(\Omega)-
\varepsilon$ and $\varphi_\alpha\in{\cal H}_0,$ $\alpha\in
\Omega_\varepsilon.$ Moreover, $\|\varphi_\alpha\|\le C,$ $\alpha
\in\Omega_\varepsilon$ for some $C>0.$
\end{lem}
\begin{proof} Let $f\in\cal H$ and $d=\| f\|^2\ne 0.$ Define two
sequences of measurable subsets in $\Omega$ by the following
equalities:
$$
A_n=\left\{ y: \int |f(x,y)|^2dx<n,\ y\in\Omega\right\},\quad n\in \mathbb N,
$$ $$
B_n=\left\{ y: \int |f(x,y)|^2dx\ge n,\ y\in\Omega\right\},\quad n\in \mathbb N.
$$
The sequences of subsets $\{ A_n\}$ and $\{ B_n\}$ hold the
following properties:\\
$1^o.$ $A_1\subset A_2\subset\ldots\subset A_n\subset \ldots$ and
$B_1\supset B_2\supset\ldots\supset B_n\supset \ldots;$\\
$2^o.$ $\lim\limits_{n\to\infty} A_n=\bigcup\limits_{n\in \mathbb N} A_n$ and  $\lim\limits_{n\to\infty} B_n=\bigcap\limits_{n\in \mathbb N} B_n;$  \\
$3^0.$ $\Omega=A_n\cup B_n$ and $A_n\cap B_n=\varnothing,$ $n\in
\mathbb N.$

Further, we define two bounded sequences of non-negative numbers $a_n$ and $b_n$ by
$$
a_n=\int\limits_{A_n}dy\int\limits_\Omega |f(x,y)|^2 dx\ \mbox{
and } \ b_n=\int\limits_{B_n}dy\int\limits_\Omega |f(x,y)|^2 dx.
$$

The sequences of numbers $a_n$ and $b_n$ have the properties:\\
$4^o.$ $a_n$ is increasing and $b_n$ is decreasing;\\
$5^o.$ $a_n+b_n=d,$ $n\in \mathbb N.$

>From the boundeness and monotonicity of the sequences $a_n$ and $b_n$ we infer that both of them have finite
limit. By the property $5^o$ and by  the construction of the set
$B_n$ we obtain that $d-a_n\ge 0,$ $n\in \mathbb N$ and $d\ge
a_n+n\mu(B_n),$ $n\in \mathbb N.$ Then $\mu(B_n)\le (d-a_n)/n,$
$n\in \mathbb N.$ Therefore $\lim\limits_{n \to
\infty}\mu(B_n)=0.$ By the property $3^o$ we have $\mu(A_n)=\mu(
\Omega)-\mu(B_n),$ $n\in \mathbb N.$ Hence, $\lim\limits_{n \to
\infty} \mu(A_n)=\mu(\Omega),$ i.e. for an arbitrary small
$\varepsilon>0$ there exists a number $n_0\in \mathbb N$ such that
$\mu(\Omega)- \varepsilon\le \mu(A_{n_0})\le \mu(\Omega)$ and
$0\le \mu(B_{n_0})< \varepsilon.$ Moreover, this means that
$$
\int|\varphi_\alpha(x)|^2dx=\int|f(x,\alpha)|^2dx< n_0,\quad
\alpha\in A_{n_0}.
$$

Then, for the set $\Omega_\varepsilon=A_{n_0}$ we have $\varphi _\alpha\in{\cal H}_0,$ $\alpha\in
\Omega_\varepsilon$ and $\| \varphi_\alpha\|\le C,$ $\alpha\in \Omega_\varepsilon$ for all $C\ge n_0.$
\end{proof}
\begin{cor} Let $f\in\cal H,$ $\| f\|=1$ and $\varphi_y(x)=f(x,y)
,$  where $y\in \Omega$ is fixed. Then there exists a measurable
subset $\Omega_0\subset \Omega$ such that, $\mu( \Omega_0)>0$ and
the family $\{\varphi_\alpha\}_{\alpha\in \Omega}$ of functions on
$\Omega$ has the following property: $\varphi_\alpha \in{\cal
H}_0,$ $\alpha\in \Omega_0$ and $0< \|\varphi_\alpha\| \le C,$
$\alpha\in \Omega_0$ for some $C>0.$
\end{cor}
\begin{cor}  Let $f\in\cal H.$ Then there exists a decreasing
sequence $\{ \varepsilon_n\}_{n\in \mathbb N}$ of posity numbers
such that $\lim\limits_{n\to\infty}\varepsilon_n=0$ and

(a) for each $n\in \mathbb N$ there exists a measurable subset
$\Omega_n \subset \Omega$ with the propertie $\mu(\Omega_n)
>\mu(\Omega)-\varepsilon_n$ such that $\Omega_1\subset\Omega_2\subset\ldots \subset\Omega_n\subset\ldots$ and $\bigcup\limits_{n\in \mathbb N}\Omega_n =\Omega;$

(b) for each $n\in \mathbb N,$ $\varphi_\alpha^{(n)}\in {\cal
H}_0,$ $\alpha\in\Omega_n$ and there exists a positive number
$C_n$ such that $\|\varphi_\alpha^{(n)}\|\le C_n,$ $\forall
\alpha\in \Omega_n,$ where $\varphi_\alpha^{(n)}(x)=f(x,\alpha),$
$\alpha\in \Omega_n;$

(c) for any $n\in \mathbb N,$ the function
$$
f_n(x,y)=\left\{\begin{array}{cl} f(x,y), & \mbox{ if }\
(x,y)\in \Omega\times \Omega_n,\\ 0,& \mbox{ otherwise } \\
\end{array} \right.
$$
belongs to $\cal H$ and $\lim\limits_{n\to\infty}f_n(x,y)=f(x,y).$
\end{cor}
\begin{pr} The following two conditions are equivalent:

(i) A number $\lambda\in\mathbb C$ is an eigenvalue for the operator $T_1;$

(ii) A number $\lambda\in\mathbb C$ is an eigenvalue for operators
$\{K_\alpha\}_{ \alpha \in \Omega_0}$,  where $\Omega_0$ is some
subset of $\Omega$ such that $\mu(\Omega_0)>0.$
\end{pr}
\begin{proof} We start with the implication $(i)\Rightarrow (ii).$ Let
$\lambda\in\mathbb C$ be an eigenvalue of operator $T_1,$ i.e.
$T_1f_0=\lambda f_0$ for some $f_0\in{\cal H},$ $\| f_0\|=1.$ We
define $\varphi_\alpha=\varphi_\alpha(x)=f_0(x,\alpha),$ $\alpha
\in\Omega.$ Therefore, we have a family $\{\varphi_\alpha\}_{
\alpha\in\Omega}$ of functions on $\Omega.$ Then, by Corollary 1,
there exists a subset $\Omega_0\subset\Omega$ such that $\mu(
\Omega_0)>0$ and $\varphi_\alpha\in{\cal H}_0,$ $\alpha\in
\Omega_0,$ $\| \varphi_\alpha\|\ne 0,$ $\forall \alpha\in
\Omega_0.$ For an arbitrary $\alpha\in \Omega_0$ we have
$$
K_\alpha\varphi_\alpha=\int k(x,s,\alpha)\varphi_\alpha(s)ds= \int
k(x,s,\alpha)f_0(s,\alpha)ds=\lambda f_0(x,\alpha)=\lambda
\varphi_\alpha(x),
$$
i.e. the number $\lambda$ is an eigenvalue for $K_\alpha,$ $\alpha
\in \Omega_0.$

Now, we prove the implication $(ii)\Rightarrow (i).$ Suppose that
there exists a subset $\Omega_0$ in $\Omega$ with
$\mu(\Omega_0)>0$ and a number $\lambda\in\mathbb C$ is an
eigenvalue for operators $K_\alpha,$ $\alpha\in \Omega_0.$ Since
$K_\alpha$ is a compact operator for all $\alpha\in \Omega,$ then
there exists a function $f_0\in {\mathcal H}$, $f_0\ne 0$ [9] such
that $T_1f_0=\lambda f_0.$
\end{proof}
\begin{pr}
If $\lambda\in\mathbb C$ is an eigenvalue of the operator $T_1,$
then the number $\overline \lambda$ is an eigenvalue of the
operator $T_1^*.$
\end{pr}
\begin{proof}
Let  $\lambda\in\mathbb C$ be an eigenvalue of the operator $T_1.$
Then there exists a subset $\Omega_0 \subset \Omega,$
$\mu(\Omega_0)>0$ such that $\lambda$ is an eigenvalue of the
every compact operator $K_\alpha,$ $\alpha\in \Omega_0.$ Therefore
the number $\overline \lambda$ is an eigenvalue of every operator
$K_\alpha^*,$ $\alpha\in \Omega_0:$
$$
K_\alpha^*\varphi=\int\overline{k(s,x,\alpha)} \varphi(s) ds,
\quad \varphi\in {\mathcal H}_0.
$$
By Proposition 1, the number $\overline \lambda$ is an eigenvalue of the adjoint operator $T_1^*.$
\end{proof}
\begin{pr} Every eigenvalue of the operator $T_1$ has infinite
multiplicity. \end{pr}
\begin{proof} Let $\lambda\in\mathbb C$ be an eigenvalue of
$T_1.$ Hence, there exists an element $f_0\in{\cal H},$ $\| f_0 \|
=1$ such that $T_1f_0=\lambda f_0.$ We consider a subspace $L_0
\subset{\cal H}:$ $L_0=\{\widetilde f\in{\cal H}: \widetilde f
(x,y)=b(y)f_0(x,y),$ where $b=b(y)$ is an arbitrary bounded
measurable function on $\Omega\}.$ For every $\widetilde f \in
L_0$ we have $T_1\widetilde f=\lambda\widetilde f,$ i.e. $L_0
\subset M_\lambda,$ where $M_\lambda$ is the eigen-subspace
corresponding to $\lambda.$ But, the subspace $L_0$ is infinite
dimensional, therefore, $M_\lambda$ is also infinite dimensional
subspace of $\cal H.$
\end{proof}

Now we consider the equation
\begin{equation}
f-\varkappa T_1f=g_0,
\end{equation}
in the space $\cal H$, where $f$ is an unknown function from $\cal H$, $g_0\in\cal H$ is given (known) function,
$\varkappa\in\mathbb C$ is a parameter of the equation, $T_1$ is PIO with a kernel $k(x,s,y)$ continuous on
$\Omega^3.$

It is clear that, if $k(x,s,y)\in C(\Omega^3)$  then for all
$\alpha\in \Omega$ the integral operators $K_\alpha$ on ${\cal
H}_0$ are compact. For each $\alpha\in\Omega$ we denote by
$\Delta_\alpha^{( 1)}(\varkappa)$ and
$M_\alpha^{(1)}(x,s;\varkappa),$ respectively, the Fredholm
determinant and the Fredholm minor of the operator $E- \varkappa
K_\alpha,$ $\varkappa\in\mathbb C$ [10], where $E$ is the identity
operator in ${\cal H}_0.$ According to the continuity of the
kernel $k(x,s,y)$ and uniform convergence of the series for
$\Delta_\alpha^{( 1)}(\varkappa)$ and $M_\alpha^{(1)}(x,s;
\varkappa)$ for every $\varkappa\in\mathbb C$ we obtain [10] that
the function $D_1(y)=D_1(y;\varkappa)$ on $\Omega$ and the
function $M_1(x,s,y)=M_1(x,s,y;\varkappa)$ on $\Omega^3,$ which
are given respectively by the equalities
$$
D_1(y;\varkappa)=\Delta_y^{( 1)}(\varkappa),\,\  y\in \Omega\ \mbox{ and } \ M_1(x,
s,y;\varkappa)=M_y^{(1)}(x,s;\varkappa),\,\ y\in\Omega,
$$
are continuous functions on $\Omega$ and $\Omega^3$ for every
$\varkappa\in\mathbb C.$

The continuous function $D_1(y)=D_1(y;\varkappa)$ ($M_1(x,s,y)= M_1( x,s,y;\varkappa)$) is called {\it a
determinant} ({\it a minor}) of the operator $E-\varkappa T_1,$ $\varkappa\in\mathbb C.$

{\bf Definition 1.} If for a number $\varkappa_0\in\mathbb C$ $\ D_1 (y;\varkappa_0)\ne 0$ for all $y\in \Omega,$
then $\varkappa_0$ is called {\it a regular number} of the PIE (1). A set of all regular numbers of the PIE (1) is
denoted by ${\cal R}_{T_1}.$

{\bf Definition 2.} If for a number $\varkappa_0\in\mathbb C$
there exists a point $y_0\in \Omega$ such that $D_1(y_0;
\varkappa_0)=0,$ then $\varkappa_0$ is called {\it a singular
number} of the PIE (1). A set of all singular numbers of the the
PIE (1) is denoted by ${\cal S}_{T_1}.$

{\bf Definition 3.} If for a number $\varkappa_0\in\mathbb C$
there exists a measurable subset $\Omega_0\subset \Omega$ with
$\mu(\Omega_0)>0$ such that $D_1(y;\varkappa_0)=0,$ $\forall y\in
\Omega_0,$ then $\varkappa_0$ is called {\it a characteristic
number} of the the PIE (1). A set of all characteristic numbers of
the PIE (1) is denoted by ${\cal X}_{T_1}.$

{\bf Definition 4.} A number $\varkappa_0\in\mathbb C$ is called {\it an essential number} of the PIE (1) if
$\varkappa_0\in {\cal S} _{T_1}\setminus{\cal X}_{T_1}.$ A set of all essential numbers of the PIE (1) is denoted
by ${\cal E}_{T_1}.$

Thus, for a parameter $\varkappa$ of the PIE (1), there exist subsets ${\cal R}_{T_1},$ ${\cal S}_{T_1},$ ${\cal
X}_{T_1},$ and ${\cal E}_{T_1}$ in $\mathbb C,$ which have the following relations:

(i) ${\cal R}_{T_1}\cup {\cal S}_{T_1}=\mathbb C$ and ${\cal
R}_{T_1}\cap {\cal S}_{T_1}=\varnothing;$

(ii) ${\cal X}_{T_1}\cup {\cal E}_{T_1}={\cal S}_{T_1}$ and ${\cal
X}_{T_1}\cap {\cal E}_{T_1}=\varnothing.$

>From Definitions 1, 2, 3 and 4 one gets that for an arbitrary
non-zero PIO $T_1$ sets ${\cal R}_{T_1}$ and ${\cal E}_{T_1}$ are
non-empty, but ${\cal X}_{T_1}$ may be empty. For example,
consider a PIE in the space $L_2([0,1]^2):$
$$
f(x,y)-\varkappa\int\limits_0^1e^{x-s}e^yf(s,y)ds=g_0(x,y),
$$
where $f$ is an unknown function in $L_2([0,1]^2),$ $g_0\in L_2(
[0,1]^2)$ is an arbitrary given function. For this PIE, the
determinant has a simple form $D_1(y;\varkappa)=1-\varkappa e^y,$
$y\in [0,1].$ Therefore ${\cal S}_{T_1} =\left[  e^{-1},1\right]$
and ${\cal X}_{T_1}=\varnothing.$

>From Proposition 1 and Definition 3 it follows
\begin{thm} A number $\lambda\in\mathbb C,$ $\lambda\ne 0,$ is an eigenvalue of
the operator $T_1$ if and only if $\lambda^{-1}\in {\cal X}_{T_1} .$ \end{thm}
\begin{thm} a) if $\varkappa_0\in {\cal E}_{T_1},$ then $\overline{
\varkappa_0}\in {\cal E}_{T^*_1};$

b) if $\varkappa_0\in {\cal X}_{T_1},$ then $\overline{
\varkappa_0} \in {\cal X}_{T_1^*}.$
\end{thm}
\begin{proof} Let $\varkappa_0\in {\cal E}_{T_1}.$ Then there
exists a point $y_0\in\Omega$ with $D_1(y_0;\varkappa_0)=0$ and we
have $\mu\{ y\in \Omega: D_1(y;\varkappa_0)=0\}=0.$ But using a
property of the determinant $D_1(y;\varkappa_0)$ we obtain that
$\overline{D_1(y;\varkappa_0)}= {\widetilde D_1(y; \overline
\varkappa_0)},$ where $\widetilde D_1(y;\overline \varkappa_0)$ is
a determinant of the operator $E-\overline \varkappa_0 T_1^*.$
Therefore, we have $\widetilde D_1(y_0; \overline \varkappa_0)=0$
and $\mu\left\{ y\in \Omega: \widetilde D_1(y;\overline
\varkappa_0)=0\right\}=0,$ i.e. the number $\overline \varkappa_0$
is an essential number of the adjoint equation $f-\overline
\varkappa_0 T_1^*f =g_0,$ and the proof of property a) is
complete. The proof of the property b) can be proceeded
analogously. \end{proof}
\begin{thm} If $\varkappa_0\in {\cal R}_{T_1}$ then for every $g_0
\in\cal H$ the PIE (1) has a unique solution on $\cal H$ and it is of the form $f=g_0+ \varkappa_0Bg_0$, where an
operator $B=B(\varkappa_0)$ acts in $\cal H$ by the formula
\begin{equation}
Bg=\int\frac{M_1(x,s,y;\varkappa_0)}{D_1(y;\varkappa_0)} g(s,y)ds,
\quad g\in{\cal H},
\end{equation}
but the corresponding homogeneous equation $f-\varkappa_0T_1f=0$ has only trivial solution (zero solution). Here
$D_1(y; \varkappa_0)$ and $M_1(x,s,y;\varkappa_0)$ are the determinant and the minor of the operator
$E-\varkappa_0T_1,$ respectively.
\end{thm}
\begin{proof} Let $\varkappa_0\in{\cal R}_{T_1}$ and $\varkappa_0
\ne 0.$ First, we prove that PIE (1) is solvable in $\cal H.$ By
Corollary 2, for the function $g_0$ there exists a decreasing
sequence of non-negative numbers $\varepsilon_n$ and a sequence of
increasing measurable subsets $\Omega_n\subset \Omega,$ which
satisfy the properties (a), (b) and (c) with $\lim\limits_{n\to
\infty}\varepsilon_n=0.$ For every $\Omega_n$ we define a subspace
$L_2^{(n)}=L_2^{(n)}(\Omega\times\Omega)$ as follows: a function
$\widetilde f\in {\cal H}$ belongs to the subspace $L_2^{(n)},$ if
it satisfies the following conditions:

(i) $\varphi_\alpha^{(n)}(x)=\widetilde f(x,\alpha)\in{\cal H}_0,$
$\forall \alpha\in \Omega_n;$

(ii) there exists a positive number $C_n$ such that $\| \varphi
_\alpha^{(n)}\|\le C_n,$ $\forall \alpha\in \Omega_n;$

(iii) $\widetilde f (x,y)=0$ if $(x,y)\in \Omega\times(\Omega \setminus \Omega_n).$

For every $f\in\cal H,$ there exists a sequence $f_n\in L_2^{(n)}
,$ $n\in \mathbb N,$ such that $\lim\limits_{n\to\infty} f_n=f.$
Therefore, first we  find a solution of the equation (1) in the
space $L_2^{(n)}$ and we can find a solution of the equation (1)
in the space $\cal H$ as the limit
$f(x,y)=\lim\limits_{n\to\infty} \widetilde f_n(x,y),$ where
$\widetilde f_n$ are solutions of the equation (1) in the space
$L_2^{(n)}.$ Thus, the equation (1) in $L_2^{(n)}$ reduces to the
following one:
\begin{equation}
\widetilde f_n(x,y)-\varkappa_0T_1\widetilde f_n(x,y)=g_n(x,y),
\end{equation}
where $g_n$ is an element of $L_2^{(n)}$ corresponding to the
function $g_0(x,y).$

Hence, by the property (b) of Corollary 2, for each fixed $y\in\Omega,$ the equation (3) reduces to the following
second type Fredholm integral equation in ${\cal H}_0:$
$$
\varphi_\alpha^{(n)}(x)-\varkappa_0 K_\alpha \varphi_\alpha^{(n)}
(x)=h_\alpha^{(n)}(x),\quad \alpha\in\Omega \eqno (3')
$$
where $\varphi_\alpha^{(n)}(x)=\widetilde f_n(x,\alpha)$ is an
unknown function in ${\cal H}_0,$ $h_\alpha^{(n)}(x)=g_n(x,
\alpha)$ is a given function in ${\cal H}_0.$

By the first fundamental Fredholm theorem, the equation $(3')$ for
every $\alpha\in\Omega_n$ has the only solution
$$
\varphi_\alpha^{(n)}=\varphi_\alpha^{(n)}(x)= h_\alpha^{(n)}(x)+
\varkappa_0 B_\alpha h_\alpha^{(n)}(x), \
$$
where the operator $B_\alpha=B_\alpha(\varkappa_0)$ acts  in
${\cal H}_0$ by the formula
$$
B_\alpha\varphi=\int\frac{M_\alpha^{(1)}(x,s;\varkappa_0)}{ \Delta
_\alpha^{(1)}(\varkappa_0)} \varphi(s)ds,(\alpha \in \Omega_n)
$$
and $B_\alpha$ is compact. Here $\Delta
_\alpha^{(1)}(\varkappa_0)$ and $M_\alpha ^{(1)}(x,s;\varkappa_0)$
are the Fredholm determinant and the Fredholm minor of the
operator $E-\varkappa_0K_\alpha,$ respectively.

It is clear, that if $\alpha\in\Omega\setminus \Omega_n$ then the
equation (3') has the solution $\varphi_\alpha^{(n)}(x) =0.$
Hence, the function $\widetilde f_n(x,y)=\varphi_y^{(n)}(x)$
belongs to the subspace $L_2^{(n)}$ and it is a solution of the
equation (3), where $\varphi_{\alpha}^{(n)}(x), \ \varphi \in
\Omega$ the solutions of the equation (3'). We define the function
$f_0\in{\cal H}$ by the equality
$f_0(x,y)=(E+\varkappa_0B)g_0(x,y),$ where the operator
$B=B(\varkappa_0)$ acts in $\cal H$ by the formula (2) and it is a
bounded operator. But, if $y\in\Omega_n$ then we have
$$
f_0(x,y)=g_0(x,y)+ \varkappa_0Bg_0(x,y) =g_n(x,y)+ \varkappa_0Bg_n(x,y)=
$$ $$
=h_y^{(n)}(x)+\varkappa_0B_yh_y^{(n)}(x)=\varphi_y^{(n)}(x)=\widetilde f_n(x,y),
$$
and for every $y\in\Omega\setminus\Omega_n$ we have
$f_0(x,y)=\varphi_y^{(n)}(x)= 0.$ Thus, by the property (c) of
Corollary 2 we obtain $f_0(x,y)=\lim\limits_{n\to\infty}\widetilde
f_n(x,y).$ Therefore the function
$f(x,y)=f_0(x,y)=(E+\varkappa_0B)g_0(x,y)$ is a solution of the
equation (1).

Thus, we have proved that the equation (1) is solvable. Now we
prove uniqueness of the solution of the equation (1). Suppose,
$f_1\in \cal H$ and $f_2\in\cal H$ are solutions of the equation
(1), where $f_1\ne f_2.$ Then, for the function $\widehat f=
f_1-f_2\ne 0$ we have $\widehat f-\varkappa_0T_1\widehat f=0,$
i.e. the homogeneous equation $f-\varkappa_0T_1f= 0$ has a
solution $\widehat f\ne 0.$ Hence, the number $\varkappa_0^{-1}$
is an eigenvalue of $T_1,$ then by Theorem 1 we obtain that
$\varkappa_0 \in{\cal X}_{T_1}.$ But this is impossible since
$\varkappa_0\in{\cal R}_{ T_1}.$

Using Proposition 1 we can show that for $\varkappa_0\in{\cal
R}_{T_1}$ the homogeneous equation $f-\varkappa_0 T_1f= 0$ has
only trivial solution. The proof is complete.\end{proof}

\begin{thm} Let $\varkappa_0\in{\cal E}_{T_1}.$ If the free term $g_0$ of the PIE (1) satisfies the
condition
$$
\int\frac{\int |g_0(s,y)|^2ds}{|D_1(y;\varkappa_0)|^2}dy<\infty, \eqno (II)
$$
then PIE (1) has a unique solution on $\cal H$ and it has a form $f=g_0+\varkappa_0Bg_0\in{\cal H},$ but
corresponding homogeneous equation $f-\varkappa_0T_1f=0$ has only trivial solution, where the operator $B$ is
given by (2).\end{thm}

\begin{proof} Let $\varkappa_0\in{\cal E}_{T_1}.$ Put $\Omega'= \{ y\in \Omega: D_1(y;\varkappa_0)=0\}.$ It is
evident that $\Omega'\ne\varnothing$ and $\mu(\Omega')=0.$
However, for every $y\in\Omega\setminus \Omega'$ the function
$f_0(x,y)=g_0(x,y)+\varkappa_0Bg_0(x,y)$ satisfies the equation
(1). Now it is enough to show that $f_0 \in {\cal H}.$ Suppose
that $g_0$ satisfies the condition (II). We have
$$
\int\int |Bg_0(x,y)|^2dxdy=\int\int\left|\int\frac{M_1(x,s,y;\varkappa_0)}{D_1(y;\varkappa_0)} g_0(s,y)ds
\right|^2dxdy\le
$$ $$
\le\int\int\left(\frac{\int|M_1(x,s,y;\varkappa_0)|\cdot |g_0(s,y)|ds}{|D_1(y;\varkappa_0)|} \right)^2 dxdy \le
$$ $$
\le N_0^2 \int\int \frac{\left(\int|g_0(s,y)|ds\right)^2}{|D_1(y;\varkappa_0)|^2}dxdy\le
$$ $$
\le N_0^2\mu(\Omega) \int \frac{\left(\int |g_0(s,y)|ds\right)^2}{|D_1(y;\varkappa_0)|^2}dy,
$$
where $N_0=\max\limits_{x,s,y\in\Omega} |M_1(x,s,y;\varkappa_0)|.$

But for the function $g_0(x,y)$ from the Cauchy-Schwartz
inequality for almost all $y\in \Omega$ we have
$$
\int|g_0(s,y)|ds\le\sqrt{\mu(\Omega)}\cdot\sqrt{\int|g_0(s,y)|^2ds}.
$$
Hence, we obtain
$$
\int\int|Bg_0(x,y)|^2dxdy\le \left( N_0\cdot\mu(\Omega)\right)^2\cdot\int \frac{\int |g_0(s,y)|^2ds}{| D_1(y;
\varkappa_0)|^2}dy<\infty,
$$
i.e. $Bg_0\in\cal H,$ therefore $f_0=g_0+\varkappa_0Bg_0\in \cal H$ and $f_0$ is a solution of the equation (1).

Uniqueness of the solution follows from Theorem 1. Using
Proposition 1 one can also show that the homogeneous equation
$f-\varkappa_0T_1f=0$ has only the trivial solution.
\end{proof}

{\bf Remark 1.} The condition (II) in Theorem 4 is natural.

For example, for the equation
\begin{equation}
f(x,y)-\varkappa \int\limits_0^1e^{x-s}yf(s,y)ds=e^xy^{1/2}
\end{equation}
in the space $L_2([0,1]^2),$ we have $D_1(y;\varkappa)=1-\varkappa
y,$ $y\in [0,1]$ and $M_1(x,s,y;\varkappa)= e^{x-s}y.$ Hence,
${\cal S}_{T_1}={\cal E}_{T_1}=[1,\infty).$ For each
$\varkappa\not\in [1,\infty),$ the equation (4) has the solution
\begin{equation}
f_0(x,y)=\frac{e^xy^{1/2}}{1-\varkappa y}\in L_2([0,1]^2).
\end{equation}
If $\varkappa_0\in[1,\infty),$ then the function (10) is a
continuous function on the set $\Omega'=[0,1]\times
([0,1]\setminus  \{ 1/\varkappa_0\})$ with
$\widehat\mu(\Omega')=\widehat \mu([0,1]\times[0,1])$ and for
every $y\in [0,1]\setminus\{1/\varkappa_0\}$ the function (10)
satisfies the equation (4), but $f_0\not\in L_2([0,1]^2).$

{\bf Remark 2.} Let $k(x,s,y)\in C(\Omega^3).$ Then, in the case of $\varkappa_0\in{\cal E}_{T_1},$ the set of all
functions $g\in \cal H$ (see Theorem 4), which satisfies the inequality
$$
\int\frac{\int|g(s,y)|^2ds}{|D_1(y;\varkappa_0)|^2}dy<\infty,
$$
is infinite dimensional subspace in $\cal H$.

\end{document}